\documentclass[conference, 10pt, twocolumn]{IEEEtran}

\usepackage{graphicx}
\usepackage{amsfonts}
\usepackage{amssymb}
\usepackage{amsmath}

\usepackage{algorithm}
\usepackage{algpseudocode}
\usepackage{algorithmicx}

\newcommand{\comment}[1]{}

\newtheorem{theorem}{Theorem}[section]
     \newtheorem{lemma}[theorem]{Lemma}

     \newcommand{\qed}{\nobreak \ifvmode \relax \else
           \ifdim\lastskip<1.5em \hskip-\lastskip
           \hskip1.5em plus0em minus0.5em \fi \nobreak
           \vrule height0.75em width0.5em depth0.25em\fi}

\begin{document}

\title{On Coding Efficiency  for Flash Memories}

\author{
\authorblockN{Xudong Ma\\}
\authorblockA{Pattern Technology Lab LLC, U.S.A.  \\
Email: xma@ieee.org}}

\maketitle

\comment{
\begin{abstract}
This paper provides the instructions for the preparation of papers
for submission to ISIT 2007 and relevant style file that produced
this page.
\end{abstract}

\section{Introduction}
The 2007 IEEE International Symposium on Information Theory will be held
at the Acropolis Congress and Exhibition Center in Nice, France,
from Sunday June 24 through Friday June 29, 2007.

\section{Submission and Review Process}
Papers will be reviewed on the basis of a manuscript of
sufficient detail to permit reasonable evaluation.
The manuscript should
{\bf not exceed five double-column pages,
with single line spacing, main text
font size no smaller than 10 points,
and at least 3/4 inch margins (about 18 mm)}.
The deadline for submission is {\bf January 8, 2007}, with
notification of decisions by {\bf March 25, 2007}.

The deadline and five page limit will be strictly
enforced. In view of the large number of submissions expected,
multiple submissions by the same author will receive especially
stringent scrutiny. All accepted papers will be allowed twenty
minutes for presentation.

\section{Proceedings}
Accepted papers will be published in full (up to
five pages in length) on CD-ROM. A hard-copy book of abstracts will
also be distributed at the Symposium to serve as a guide to the
sessions.

The deadline for the submission of the final camera-ready
paper is {\bf April 20, 2007}. Final manuscript guidelines will be
made available after the notification of decisions.

\section{Preparation of the Paper}
Only electronic submissions in form of a Postscript (PS) or Portable Document
Format (PDF) file will
be accepted. Most authors will prepare their papers with \LaTeX. The
\LaTeX\ style file (\verb#IEEEtran.cls#) and the \LaTeX\ source
(\verb#isit_example.tex#) that produced this page may be downloaded
from the ISIT 2007 web site (http://www.isit2007.org).
Do not change the style file in any way.
Authors using other means to prepare their manuscripts should attempt
to duplicate the style of this example as closely as possible.


The style of references, e.g.,
\cite{Shannon1948}, equations, figures, tables, etc., should be the
same as for the \emph{IEEE Transactions on Information Theory}. The
affiliation shown for authors should constitute a sufficient mailing
address for persons who wish to write for more details about the
paper.

\section{Electronic Submission}
The paper submission portal is EDAS:
\begin{verbatim}
http://edas.info/newPaper.php?c=5045&
\end{verbatim}
Check ISIT2007 website (http://www.isit2007.org) for relevant
announcements.

\section{Conclusion}
Never conclude a real information theoretic paper.
If you have to, the conclusion goes here.


\section*{Acknowledgment}
The authors would like to thank various sponsors for supporting their research.
In particular, we thank the TPC chairs of ISIT 2006 for
providing the \LaTeX\ templates for paper submission.



%
}

\begin{abstract}

Recently, flash memories have become a competitive solution for mass
storage. The flash memories have rather different properties
compared with the rotary hard drives. That is, the writing of flash
memories is constrained, and flash memories can endure only limited
numbers of erases. Therefore, the design goals for the flash memory
systems are quite different from these for other memory systems. In
this paper, we consider the problem of coding efficiency. We define
the ``coding-efficiency'' as the amount of information that one
flash memory cell can be used to record per cost. Because each flash
memory cell can endure a roughly fixed number of erases, the cost of
data recording can be well-defined. We define ``payload'' as the
amount of information that one flash memory cell can represent at a
particular moment. By using information-theoretic arguments, we
prove a coding theorem for achievable coding rates. We prove an
upper and lower bound for coding efficiency. We show in this paper
that there exists a fundamental trade-off between ``payload'' and
``coding efficiency''. The results in this paper may provide useful
insights on the design of  future flash memory systems.

\end{abstract}

\section{Introduction}

Recently, flash memories have become a competitive solution for mass
storage. Compared with the conventional rotary hard drives, flash
memories have high random access read speed, because there is no
mechanical seek time. Flash memory storage devices are also more
lightweight, power efficient, and kinetic shock resistant.
Therefore, they are becoming desirable choices for many applications
ranging from high-speed servers in data centers to portable devices.

Flash memories are one type of solid state memories. Each piece of
flash memory usually contains multiple arrays of flash memory cells.
Each memory cell is a transistor with a floating gate. Information
is recorded using one memory cell by injecting and removing
electrons into and from the floating gate. The process of injecting
electrons is called programming and the process of removing
electrons is called erase. Programming increases the threshold
voltage level of the memory cell, while erase decreases the
threshold voltage level. The threshold voltage level of the memory
cell is the voltage level at the control gate that the transistor
becomes conducting. In the reading process for the memory cell, the
threshold voltage level is detected, thus the recorded information
can be recovered.

The memory cells are organized into pages and then into blocks. The
programming is page-wise and erase is block-wise. Usually, one
memory block is first erased, so that all memory cells within the
block return to an initial threshold voltage level. After the erase
operation, the pages in the block are programmed (possibly multiple
times), until normal threshold voltage level ranges are used up.
Then, the memory block is erased again for further use.

One challenge for flash memories is that the number of erase
operations that one memory cell can withstand is quite limited. For
 current commercial flash memories, such maximal numbers of block
erase operations range from 5,000 to 100,000. After such a limited
number of erase operations, the flash memory cell would become
broken or unreliable. Therefore, data encoding methods must be
carefully designed to address such an issue.

In fact, flash memories can be considered as one type of
write-once-memories. The write-once-memories were first discussed in
the seminal work by Rivest and Shamir \cite{rivest82}. Previous
examples of write-once-memories include digital optical disks,
punched paper tapes, punched cords, and programmable read-only
memories etc. Rivest and Shamir show that by using advanced data
encoding methods, the write-once-memories can be rewritten. In
\cite{rivest82}, one theorem for the achievable data recording rates
of binary write-once-memories has been proven using combinatorial
arguments. During the passed research, many data encoding methods
for rewriting the write-once-memories have been proposed, see for
example, \cite{jinag09} \cite{wu11} etc.

In this paper, we consider a coding efficiency problem for data
encoding on flash memories. Unlike other type of computer memories,
the cost of data encoding can be well-defined for flash memories.
That is, the cost for each erase operation can be defined based on
the cost of the flash memory block and the total number of erase
operations that the memory block can have. The coding efficiency
problem is therefore the problem of recording more information using
fewer erase operations. To our best knowledge, such a design problem
for flash memories has never been discussed before.

We assume that one flash memory block has $N$ cells, and each cell
can take $K$ voltage levels. We assume that the data encoding scheme
uses the memory block for $T$ rounds between two consecutive erase
operations. That is, in the first round, a message ${ M}[1]$ is
recorded using the block, and in the second round, a new message
$M[2]$ is recorded, and so on. Suppose that $Nl_t$ bits are recorded
during the $t$-th round. We define the payload $p$ and coding
efficiency $c$ as
\begin{align}
p = \frac{1}{T} \sum_{t=1}^{T} l_t, \,\,\, c = \frac{\alpha}{K}
\sum_{t=1}^{T} l_t,
\end{align}
where, $\alpha$ is a constant depending on the type of the memory
block, e.g., NOR type, NAND type, single-level-cell,
multi-level-cell etc. The constant $\alpha$ may be used to reflect
the cost for the flash memory block. It should be clear that the
coding efficiency measures the amount of recording information per
voltage level cost. We may also define the voltage level cost per
recorded bit, which is exactly $1/c$.

In this paper, we first prove a coding theorem for achievable rates
of data encoding on flash memories using information-theoretic
arguments.  Using the coding theorem in this paper, we prove an
upper bound for the optimal coding efficiency. We also show a lower
bound of optimal coding efficiency using a specific coding scheme.
Surprisingly, we find that there exists a tradeoff between the
optimal coding efficiency and payload. These results may provide
useful insights and tools for designing future flash memory storage
systems.

The rest of this paper is organized as follows. In Section
\ref{section_coding_theorem}, we present the coding theorem for
achievable coding rates. In Section
\ref{section_entropy_upper_bound}, we show the upper bound of the
optimal coding efficiency. In Section \ref{section_lower_bound}, we
present the lower bound for optimal coding efficiency using a
specific coding scheme. The coding efficiency to payload tradeoff is
discussed in Section \ref{section_tradeoff}. Some concluding remarks
are presented at Section \ref{sec_conclusion}.

\section{Coding Theorem}
\label{section_coding_theorem}

We consider a memory block with $N$ memory cells. Each memory cell
can take $K$ threshold voltage levels, that is, each memory cell can
be at one of the states $0,1,\ldots,K-1$. After one erase operation,
all memory cells are at the state $K-1$. During each programming
process, the state of each cell can be decreased but never
increased. Assume that the memory block can be reliably used for $T$
rounds of information recording, where messages
$M(1),M(2),\ldots,M(t),\ldots,M(T)$ are recorded. We define the
corresponding data rate in the $t$-th round  $l(t) =
\log_2(|M(t)|)/N$, where $|M(t)|$ denote the alphabet size of the
message $M(t)$. In this case, we say that the sequence of data rates
$l(t)$, $t=1,\ldots,T$ is achievable. We assume that all the $T$
messages are statistical independent. We denote the state of the
$n$-th cell in the block during time $t$ by $X_n(t)$. We use the
notation $X_1^N(t)$ to denote the sequence $X_1(t), X_2(t),\ldots,
X_N(t)$. Similarly, $X_1^n(t)$ denotes the sequence $X_1(t),
X_2(t),\ldots, X_n(t)$, where $1\leq n \leq N$. We use $H(\cdot)$ to
denote the entropy and conditional entropy functions as in
\cite{cover}.

\begin{theorem}
\label{main_coding_theorem} A sequence of data rates $l(t)$,
$t=1,\ldots,T$ is achievable, if and only if, there exist random
variables $U(1),\ldots, U(T)$ jointly distributed with a probability
distribution ${\mathbb P}\left(U(1),\ldots, U(T)\right)$, such that,
\begin{align}
& {\mathbb P}\left(U(t)=j | U(t-1)=i\right) = 0,\mbox{ if }j>i,
\mbox{ for } t = 2, \ldots, T, \notag \\
&l(t) \leq {H}\left(U(t) | U(t-1)\right), \mbox{ for }t=2,\ldots,T, \notag \\
&l(1) \leq {H}\left(U(1) \right).
\end{align}
By convention, $U(0) = K-1$ with probability 1.
\end{theorem}

\begin{proof} The achievable part is proven by random binning. For the $t$-th
round of data recording, we construct a random code by throwing
typical sequences of $U(t)$ into $\exp\left\{Nl(t)\right\}$ bins
uniformly in random. The message $m(t)$ is encoded by finding a
sequence $X_1^N(t)$ in the $m(t)$-th bin, such that the sequence
$X_1^N(t)$ is jointly typical with $X_{1}^{N}(t-1)$. If such a
sequence can not be found,  then one encoding error is declared.

Suppose that $l(t)\leq {H}\left(U(t) | U(t-1)\right) - 2\epsilon$,
where $\epsilon$ is an arbitrarily small positive number. Then, the
probability of encoding error can be upper bounded as follows.
\begin{align}
{\mathbb P}(\mbox{error}) & = \left(1-\frac{1}{ \exp(Nl(t))
}\right)^{N_1}
\notag \\
& \stackrel{(a)}{\leq} \exp\left( - \frac{ N_1 }{ \exp(Nl(t))  } \right) \notag \\
& \stackrel{(b)}{\leq} \exp\left( -\frac{
\exp\left(N\left(H(U(t)|U(t-1))-\epsilon\right)\right)
 }{
\exp\left\{N\left(H(U(t)|U(t-1))-2\epsilon\right)\right\}
 } \right) \notag \\
& \leq \exp\left(-\exp(\epsilon N)\right)
\end{align}
where, $N_1$ denotes the number of typical sequences  $X_1^N(t)$
that are jointly typical with $X_{1}^{N}(t-1)$, (a) follows from the
inequality, $(1-x)\leq \exp(-x)$, for $0\leq x < 1$, (b) follows
from the fact that $N_1 \geq \exp\left\{N
\left(H(U(t)|U(t-1))-\epsilon\right)\right\}$. The achievable part
of the proof then follows from the fact that $\epsilon$ can be taken
arbitrarily small.

We prove the converse part by constructing some random variables
$U(1),\ldots,U(T)$, which satisfy the conditions in the theorem.
Assume that there exists at least one coding scheme, which satisfies
the conditions in the theorem.

In the first step, we wish to show
\begin{align}
 H\left(M(t)\right) \leq
H\left(X_{1}^{N}(t)\left|X_{1}^{N}(t-1)\right.\right)
\end{align}
This is because, on the one hand,
\begin{align}
& H\left(M(t),X_{1}^{N}(t)|X_{1}^{N}(t-1)\right) \notag \\
& = H\left(X_{1}^{N}(t)|X_{1}^{N}(t-1)\right) + H\left(M(t) |
X_{1}^{N}(t), X_{1}^{N}(t-1)\right)
\notag \\
& \stackrel{(a)}{=} H\left(X_{1}^{N}(t)|X_{1}^{N}(t-1)\right)
\end{align}
where, (a) follows from the fact that $M(t)$ can be completely
determined by observing $X_{1}^{N}(t)$. On the other hand,
\begin{align}
\label{lemma_entropy_bound_one} &
H\left(M(t),X_{1}^{N}(t)|X_{1}^{N}(t-1)\right) \notag \\
& = H\left(M(t)|X_{1}^{N}(t-1)\right) + H\left( X_{1}^{N}(t)| M(t),
X_{1}^{N}(t-1)\right)
\notag \\
& \stackrel{(a)}{=} H\left(M(t)\right) + H\left( X_{1}^{N}(t)| M(t),
X_{1}^{N}(t-1)\right)
\end{align}
where, (a) follows from the fact that $M(t)$ is independent of
$X_{1}^{N}(t-1)$.

In the second step, we can show that
\begin{align}
H\left(M(t)\right)\leq \sum_{n=1}^{N}
H\left(X_{n}(t)|X_{n}(t-1)\right)
\end{align}
This is because,
\begin{align}
H\left(X_{1}^{N}(t)|X_{1}^{N}(t-1)\right) & = \sum_{n=1}^{N}
H\left(X_n(t) | X_{1}^{n-1}(t), X_{1}^{N}(t-1)\right) \notag \\
& \leq \sum_{n=1}^{N} H\left(X_n(t) | X_{n}(t-1)\right)
\end{align}
where the last inequality follows from the fact that conditions do
not increase entropy.

Let us define random variables $Z, U(1), U(2), ,\ldots, U(T)$ as
follows. The random variable $Z$ takes values in $\{1,2,\ldots,N\}$
uniformly in random.
\begin{align}
U(t) = X_n(t), \mbox{ if }Z=n
\end{align}
The probability distribution of the random variables $Z, U(1), U(2),
,\ldots, U(T)$ can be factored as follows.
\begin{align}
{\mathbb P}(Z) \prod_{t=1}^{T} {\mathbb P}(U(t) | U(1), \ldots,
U(t-1),Z)
\end{align}
It can be checked that
\begin{align}
{\mathbb P}\left(U(t)=j | U(t-1)=i\right) = 0,\mbox{ if }j>i
\end{align}

Finally, we wish to  show that
\begin{align}
Nl(t) = H\left(M(t)\right)\leq NH\left(U(t)|U(t-1)\right)
\end{align}
This is because
\begin{align}
& H(M(t)) \leq \sum_{n=1}^{N} H\left(X_n(t) | X_{n}(t-1)\right)
\notag \\
& \stackrel{(a)}{=} NH\left(U(t)|U(t-1),Z\right) \notag \\
& \stackrel{(b)}{\leq} NH\left(U(t)|U(t-1)\right)
\end{align}
where, (a) follows from the definition of $Z$, (b) follows from the
fact that conditions do not increase entropy.

Therefore, we have constructed the random variables
$U(1),\ldots,U(T)$, which satisfy the conditions in the theorem. The
theorem is proven.

\end{proof}


\section{Upper Bound}
\label{section_entropy_upper_bound}

In this section, we prove an upper bound for the achievable coding
efficiency. It is clear that the coding efficiency can be calculated
based on the Theorem \ref{main_coding_theorem} by forming an
optimization problem. Let us define a random variable $V(t) = U(t-1)
- U(t)$ with an alphabet $\{0,1,\ldots,K-1\}$. With a given payload
$p$, the optimization problem is as follows.
\begin{align}
\label{original_optimization_problem} & \min_{{\mathbb
P}(V(1),\ldots,V(t), \ldots V(T))}
{\mathbb E}\left[\sum_t V(t)\right]  \\
& \mbox{Subject to: } \sum_t H(V(t)|U(t-1)) \geq  Tp \\
& \hspace{0.5in} {\mathbb P}\left(\sum_t V(t)\geq K \right) = 0
\label{prob_constr}
\end{align}
By convention, $U(0)=K-1$ with probability $1$. It should be clear
that the coding efficiency
\begin{align}
c \leq \frac{\alpha Tp}{\sum_t {\mathbb E}(V(t)^\ast)}
\end{align}
where $V(t)^\ast$ denotes the minimizer of the optimization problem.

However, the above optimization problem is difficult to solve in
closed-form. We will consider instead a relaxed optimization
problem. First, we remove the constraint in Eqn \ref{prob_constr}.
Second, we  relax the constraint $\sum_t H(V(t)|U(t-1)) \geq Tp$ to
$\sum_t H(V(t)) \geq Tp$, due to the fact that conditions do not
increase entropy. Thus, the original optimization problem becomes
\begin{align}
\label{second_optimization_problem} & \min_{{\mathbb
P}(V(1),\ldots,V(t), \ldots V(T))}
{\mathbb E}\left[\sum_t V(t)\right] \notag \\
& \mbox{Subject to: } \sum_t H(V(t)) \geq  Tp
\end{align}

In a final step, because all the constraint and objective functions
only depend on marginal distributions of $V(t)$, we may further
relax the above optimization problem by replacing the joint
distribution \begin{align}{\mathbb P}(V(1),\ldots,V(t), \ldots,
V(T))\end{align} with a set of pseudo marginal distributions
\begin{align}{\mathbb P}(V(1)),\ldots,{\mathbb P}(V(t)), \ldots, {\mathbb
P}(V(T))\end{align} The pseudo marginal distributions may or may not
correspond to a joint distribution. The final relax optimization
problem is thus as follows.
\begin{align}
\label{relax_optimization_problem} & \min_{{\mathbb P}(V(1)),\ldots,
{\mathbb P}(V(t)), \ldots, {\mathbb P}(V(T)))}
{\mathbb E}\left[\sum_t V(t)\right] \notag \\
& \mbox{Subject to: } \sum_t H(V(t)) \geq  Tp
\end{align}

Using the Lagrangian method, we can find that the optimal
distribution for $V(t)$ takes the following form
\begin{align}
\label{opt_dist_form} {\mathbb P}(V(t) = j) = \frac{\exp(-\beta_t
j)}{\sum_{s=0}^{K-1}\exp(-\beta_t s)}
\end{align}
for a certain parameter $\beta_t>0$. Let us define the cost function
$\mbox{cost}(\beta_t)$ and rate function $\mbox{rate}(\beta_t)$ at
the $t$-th data encoding round as follows.
\begin{align}
& \mbox{cost}(\beta_t) = {\mathbb E}\left[V(t)\right], \,\,\,
\mbox{rate}(\beta_t) =  H(V(t))
\end{align}
where $V(t)$ has a probability distribution in Eqn.
\ref{opt_dist_form}. Both the two functions have closed-form
formula,
\begin{align}
& \mbox{cost}(\beta_t)  = \frac{\sum_{j=0}^{K-1}
j\exp(-\beta_t j)}{\sum_{s=0}^{K-1}\exp(-\beta_t s)} \notag \\
& \mbox{rate}(\beta_t)   = \beta_t \mbox{cost}(\beta_t) +
\log\left(\sum_{s=0}^{K-1}\exp(-\beta_t s)\right)
\end{align}

\begin{theorem}
The coding efficiency $c$ is upper bounded by
\begin{align}
c \leq \frac{\alpha \sum_t \mbox{rate}((\beta_t))}{ \sum_t
\mbox{cost}( \beta_t ) }
\end{align}
where, $\beta_t$ corresponds to the solution to the relaxed
optimization problem in Eqn. \ref{relax_optimization_problem}.
\end{theorem}
\begin{proof}
The optimal value of a relaxed maximization optimization problem is
greater than or equal to the optimal value of the original
optimization problem.
\end{proof}

In our further discussion, we need to define a stage coding
efficiency function
\begin{align}
f(\beta)=\frac{\mbox{rate}(\beta)}{\mbox{cost}(\beta)}
\end{align}

\begin{lemma}
\label{rate_cost_der_lemma}
\begin{align}
\frac{\mbox{d} (\mbox{rate}(\beta))}{\mbox{d} (\mbox{cost}(\beta))}
= \beta
\end{align}
\end{lemma}
\begin{proof}
\begin{align}
& \frac{\mbox{d } \mbox{rate}(\beta)}{\mbox{d } \mbox{cost}(\beta)} \notag \\
& = \frac{ \mbox{cost}(\beta) + \beta \mbox{cost}'(\beta) + \sum_k
-k\exp(-\beta k) / \sum_s \exp(-\beta s)} {\mbox{cost}'(\beta)}
\notag \\
& = \frac{ \mbox{cost}(\beta) + \beta \mbox{cost}'(\beta)
-\mbox{cost}(\beta)} {\mbox{cost}'(\beta)} = \beta
\notag \\
\end{align}
where, the derivatives at the right hand sides are with respect to
$\beta$.
\end{proof}

\begin{lemma}
\label{cost_decrease_lemma} The function $\mbox{cost}(\beta)$ is a
decreasing function with respect to $\beta$.
\end{lemma}
\begin{proof}
In order to show that $\mbox{cost}(\beta)$ is a decreasing function,
it is sufficient to show that $\log(\mbox{cost}(\beta))$ is a
decreasing function. The derivative of $\log(\mbox{cost}(\beta))$ is
\begin{align}
\frac{\sum_{k=0}^{K-1}k\exp(-k\beta)}{\sum_{k=0}^{K-1}\exp(-k\beta)}
-
\frac{\sum_{k=0}^{K-1}k^2\exp(-k\beta)}{\sum_{k=0}^{K-1}k\exp(-k\beta)}
\end{align}
By using the Cuachy-Schwarz inequality, we have
\begin{align}
\left[\sum_{k=0}^{K-1}k\exp(-k\beta)\right]^2 \leq
\sum_{k=0}^{K-1}\exp(-k\beta)\sum_{k=0}^{K-1}k^2\exp(-k\beta)
\end{align}
and the equality holds only when $\beta$ goes to infinity. It thus
follows that  the derivative of $\log(\mbox{cost}(\beta))$ is
strictly negative for any finite $\beta$. The lemma follows.
\end{proof}

\begin{lemma}
\label{f_beta_increasing_lemma} The function $f(\beta)$ is an
increasing function with respect to $\beta$.
\end{lemma}
\begin{proof}
The derivative of $f(\beta)$ is as in Eqn. \ref{der_f_beta}.
\begin{figure*}
\begin{align}
\label{der_f_beta} f'(\beta)  =
 \log\left( \sum_{k=0}^{K-1}\exp(-k\beta) \right)
\left\{ -1 +
 \frac{
\left[\sum_{k=0}^{K-1} \exp(-k\beta)\right]
\left[\sum_{k=0}^{K-1}k^2\exp(-k\beta)\right] }{
\left[\sum_{k=0}^{K-1}k\exp(-k\beta)\right]^2 } \right\}
 \end{align}
 \end{figure*}

The lemma is proven if we can show that
\begin{align}
 \frac{
\left[\sum_{k=0}^{K-1} \exp(-k\beta)\right]
\left[\sum_{k=0}^{K-1}k^2\exp(-k\beta)\right] }{
\left[\sum_{k=0}^{K-1}k\exp(-k\beta)\right]^2 }\geq 1
\end{align}
That is,
\begin{align}
\left[\sum_{k=0}^{K-1}k\exp(-k\beta)\right]^2 \leq
\left[\sum_{k=0}^{K-1} \exp(-k\beta)\right]
\left[\sum_{k=0}^{K-1}k^2\exp(-k\beta)\right]
\end{align}
We can show that this is indeed the case by using the Cuachy-Schwarz
inequality,
\begin{align}
\left(\sum_k \sqrt{x_ky_k}\right)^2 \leq \left(\sum_k
x_k\right)\left(\sum_k y_k\right)
\end{align}
\end{proof}

\begin{theorem}
In the solution to the optimization problem in Eqn.
\ref{relax_optimization_problem},
\begin{align}\beta_1=\beta_2=\ldots = \beta_t = \ldots = \beta_T =
\beta.
\end{align}
Therefore, the coding efficiency
\begin{align}
c \leq \frac{ \alpha \mbox{rate}((\beta))}{  \mbox{cost}( \beta ) }
\end{align}
\end{theorem}
\begin{proof}
The theorem is proven by contradiction. Suppose that in the
optimization solution for Eqn. \ref{relax_optimization_problem},
there exist $\beta_s$ and $\beta_t$ such that $\beta_s>\beta_t$.
According to Lemma \ref{cost_decrease_lemma},
$\mbox{cost}(\beta_s)<\mbox{cost}(\beta_t)$. We may modifity
$\beta_s$ and $\beta_t$ slightly into $\beta_s - \Delta \beta_s$ and
$\beta_t + \Delta \beta_t$, such that
\begin{align}
\mbox{cost}(\beta_s - \Delta \beta_s) = \mbox{cost}(\beta_s) +
\Delta \mbox{cost}
\end{align}
\begin{align}
\mbox{cost}(\beta_t + \Delta \beta_t) = \mbox{cost}(\beta_t) -
\Delta \mbox{cost}
\end{align}
where $\Delta \mbox{cost}>0$. Therefore, the total sum of cost
functions remains the same. On the other hand, the rate function
corresponding to $\beta_s$ increases with derivative $\beta_s$, and
the rate function corresponding to $\beta_t$ decreases with
derivative $\beta_t$. The total sum of rate functions increases.
Therefore, $\beta_s$ and $\beta_t$ can not be a part of the
optimization solution. This results in a contradiction. The theorem
is proven.
\end{proof}


\section{Achievable Lower Bound using Random Coding Arguments}
\label{section_lower_bound}

In this section, we prove a lower bound for the coding efficiency by
using a specific random coding scheme. The data encoding scheme
consists of multiple stages. During all the stages, the cells in the
block are restricted to take one of two states, $k$ or $k-1$, where
$k=1,\ldots,K-1$. Assume in a certain stage, there are $l$ cells
that take the state $k-1$, and the rest  $N-l$ cells take state $k$.
Then, during this stage, the state of only one memory cell is
changed from $k$ to $k-1$ and $l(t)=
\log_2\lfloor(1-\epsilon)(N-l)\rfloor$ bits can be recorded, where
$\lfloor \cdot \rfloor$ denotes the floor function, and $\epsilon$
is a small real number, $0<\epsilon<1$.

The data encoding process is as follows. Let us throw all the
sequences of symbols with length $N$ and alphabet
$\{0,1,2,\ldots,K-1\}$ into $2^{(l(t))}$ bins uniformly in random.
If the to-be-recorded message is $m[t]$, then we check the $m[t]$-th
bin. We try to find one sequence in the bin, such that the current
configuration of the memory cells can be modified to be equal to the
sequence by turning the state of one memory cell $X_n$ from $k$ to
$k-1$. If such a sequence can be found, then we turn the state of
the memory cell $X_n$ from $k$ to $k-1$. If such a sequence can not
be found in the bin, then a decoding error is declared and we
randomly turn one memory cell from $k$ to $k-1$ and go to the next
coding stage.

We assume that the data decoding process knows the random coding
schemes, for example, by sharing the same random source with the
encoder, or using a pseudo random source. In the first step of data
decoding, the decoder can determine the stage of data encoding by
looking at the states of the memory cells and the number $l$ of
cells being at the state $k-1$. The recorded message $m(t)$ can then
be recovered by looking at the bin index of the current
configuration of the memory cells.

The encoding error probability can be bounded as follows,
\begin{align}
& {\mathbb P}(\mbox{error})   \leq
\left[1-\frac{1}{\lfloor(1-\epsilon)(N-l)\rfloor}\right]^{N-l}
\notag \\
& \stackrel{(a)}{\leq} \exp
\left(-\frac{N-l}{\lfloor(1-\epsilon)(N-l)\rfloor}\right)  \leq
\exp\left(\frac{-1}{1-\epsilon}\right)
\end{align}
where, (a) follows from the inequality, $ (1-x)^y \leq \exp(-xy),
\mbox{ for }x\in(0,1), y\geq 0$.

The expected total amount of recoded information between two erase
operations can be bounded as
\begin{align}
& {\mathbb E}(\mbox{rate})  \geq (K-1) \notag \\
& \times \sum_{l=0}^{N}\left[1-\exp\left (\frac{-1}{1-\epsilon}
\right)\right] \log_2\left(\lfloor(1-\epsilon)(N-l)\rfloor\right)
\end{align}
For sufficiently large $N$ and $\epsilon=0.5$, the total expected
recorded information is lower bounded as
\begin{align}
{\mathbb E}(\mbox{rate}) \geq   \frac{(K-1)N}{2} \left[1-\exp\left(
-2 \right)\right] \log\left(N/2\right)
\end{align}
Therefore, the coding efficiency is bounded as follows.
\begin{align}
c \geq   \frac{\alpha}{2} \left[1-\exp\left (-2 \right)\right]
\log\left(N/2\right)
\end{align}
The payload can be calculated as
\begin{align}
p & = \frac{1}{N^2}\sum_{l=0}^{N-1}
\log_2\left(\lfloor(1-\epsilon)(N-l)\rfloor\right)
\end{align}

Based on the above discussions in this section, we arrive at the
following theorem.
\begin{theorem}
The optimal coding efficiency for $K$ level $N$ cell flash memories
can go to infinity as $N$ goes to infinity.
\end{theorem}


\section{The Coding-Efficiency-to-Payload Tradeoff}

\label{section_tradeoff}

Some important insights can be gained from the upper and lower
bounds for coding efficiency proved in the previous sections. From
the upper bound, it can be seen clearly that the coding efficiency
decreases as the payload increases. From the lower bound, it can be
seen that the coding efficiency may go to infinity as the payload
decreases to zero. Therefore, we can conclude that there exists a
tradeoff between the coding efficiency and payload. The tradeoff is
illustrated in Fig. \ref{upper_lower_bound_fig}. In the figure, the
upper and lower bound for coding efficiency are shown, where the
x-axis shows the payload. We assume $\alpha=1$, and the flash
memories are 8-level (3bit) TLC type flash memories.

\begin{figure}[h]
 \centering
 \includegraphics[width=3in]{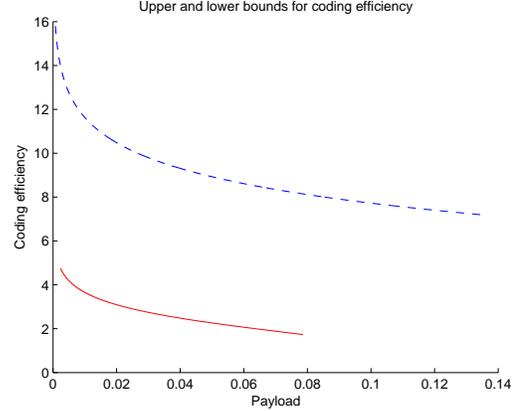}
 \caption{Upper and lower bounds for coding efficiency of 3-bit flash memory cells.}
 \label{upper_lower_bound_fig}
\end{figure}


\section{Conclusion}

\label{sec_conclusion}

In this paper, we study the coding efficiency problem for flash
memories. A coding theorem for achievable rates is proven. We prove
an upper and lower bounds for the coding efficiency. We show that
there exists a tradeoff between the coding efficiency and payload.
Our discussions in this paper provide useful insights on the design
of future flash memory systems.

\bibliographystyle{IEEEtran}
\bibliography{flash_memory_main_bib}

\end{document}